\documentclass[11pt]{article}
\usepackage{fullpage}               
\usepackage{graphicx}
\usepackage{amssymb,amsmath,amsthm}
\usepackage{epstopdf}
\usepackage[sort]{cite}

\newcommand{\cG}{\mathcal{G}}
\newcommand{\cB}{\mathcal{B}}
\newcommand{\cT}{\mathcal{T}}

\usepackage{subcaption}

\newtheorem{theorem}{Theorem}
\newtheorem{lemma}[theorem]{Lemma}

\newcommand{\E}[1]{{\mathbb E} \left [ #1 \right ]}

\DeclareMathOperator{\polylog}{polylog}
 
\title{Matchings via Random Greedy Independent Set: \\A Simpler Algorithm and Analysis}
\author{Andrew McGregor\thanks{University of Massachusetts Amherst. The author is supported by the National Science Foundation under grant CCF-2521579.}}
\date{}                                         

\begin{document}
\maketitle

\begin{abstract}
We show that a simple extension of the randomized greedy maximal independent set algorithm yields a constant approximation for the maximum matching problem. 
The algorithm is a simplification of an algorithm used by Assadi et al.~[JACM 2026] in the context of processing data streams in the dynamic setting where edges may be inserted and deleted. In contrast to the previous work, our analysis avoids consideration of fractional matchings and yields a significantly shorter and more direct proof of the approximation factor for the basic algorithm. 
\end{abstract}

\newpage

\section{Introduction} 

Given an undirected graph on $n$ vertices, the \emph{random greedy maximal independent set} (RGMIS) algorithm picks a vertex $x$ uniformly at random, removes $x$ along with its neighbors $C$, and recurses on the remaining graph, until no vertices remain. The union of all the sets of neighbors is a vertex cover and a recent result by Veldt \cite{Veldt24} shows that, in expectation, it is at most a factor 2 greater than the size of the minimum vertex cover. 

We consider a variant of this algorithm: before removing $C$, we sample a random incident edge from each vertex in $C$. More formally, and to introduce some notation, the  algorithm is as follows:
\begin{enumerate}
\item {\bf Random Greedy MIS:} Given graph $G=(V,E)$, let  $V_1\leftarrow V$ and for $i\geq 1$ let:
\[
\underbrace{x_i\in_R V_{i}}_{\mbox{$x_i$ chosen uniformly at random from $V_i$}}
\quad
\underbrace{C_i\leftarrow \Gamma(x_i) \cap V_i}_{\mbox{neighbors of $x_i$ amongst $V_i$}}
\quad
V_{i+1}\leftarrow V_i - \{x_i\}-  C_i
\]
\item {\bf Edge Sampling:} For $i\geq 1$, let $E_i=\{\{v,f(v)\}: v\in C_i\}$ be the set of sampled edges where $f(v)\in_R \Gamma(v) \cap V_i$ is chosen independently for each $v\in C_i$. 
\item {\bf Output:} Return the maximum matching amongst  edges $\{v,f(v)\}$ for $v\in \cup_{i\geq 1 } C_i$.
\end{enumerate}

A very closely related algorithm was considered by Assadi et al.~\cite{AssadiMatching26}. Their algorithm is based on the same RGMIS approach except that $O(\log n)$ edges are sampled incident to each vertex in $C_i$. They showed that a) the output is a constant approximation for the maximum matching problem in expectation, and b) if the input graph is defined by a stream of edges being inserted and deleted, then the algorithm can be implemented using $O(n \cdot \polylog n)$ space and $O(\log \log n)$ passes.\footnote{
 The sequence $x_1, C_1, x_2, C_2,\ldots $ defined in Step 1 can be computed in the insert-delete data stream model using $O(n \polylog n)$ space and $O(\log \log n)$ passes; the sequence $x_1, x_2, \ldots $ can be constructed as in Ahn et al.~\cite{AhnCGMW15} and the sets $C_i$ can also be determined in $O(\log \log n)$ passes as described by Assadi et al.~\cite{AssadiMatching26}. Step 2 can be implemented in one additional pass via $\ell_0$ sampling, and Step 3 can be performed in post-processing.}

Our contribution is to show that this more elementary version of the algorithm, in which only one edge is sampled incident to each vertex, also works, and to provide a more direct proof of an improved approximation factor. 
Specifically we prove that:

\begin{theorem}\label{thm:main} For any graph $G$,
\[\E{\mu(\cup_i E_i)}\geq  \mu(G)/27\] where $\mu(\cdot)$ is the maximum matching size.
\end{theorem}

The spirit of our analysis is similar but by avoiding the introduction of fractional matchings, the final approximation factor is  improved, the algorithm is slightly more efficient, and the proof is significantly shorter. In the context of the $O(\log \log n)$ pass data stream algorithm the improved constant approximation is not particularly important since the accuracy can be boosted by using the RGMIS-based algorithm as a sub-routine.  

\paragraph{Related Work.} 
In addition to the result of Veldt \cite{Veldt24} mentioned above, many aspects and applications of the RGMIS algorithm have been studied in the literature. For example, one application is to correlation clustering on complete graphs where each edge has either a positive or negative label.  Ailon et al.~\cite{AilonCN08}  proved that if RGMIS is applied to the positive edges,
then taking each set $\{x_i\}\cup C_i$ as a cluster results in a 3-approximation to the problem of partitioning the vertices in order to minimize the number of positive edges between clusters plus the number of negative edges within each cluster. The properties of recursive \cite{DalirrooyfardMM26,YoshidaYI12} and parallel \cite{FischerN20,BlellochFS12} implementations have been studied. In the mathematics community, the properties of the independent set generated by the RGMIS algorithm have been studied on various graph families. See  \cite{KrivelevichMMS24,GamarnikG10}, and references therein.  

Graph matchings in the data stream model have been studied for over twenty years. Recently, two central problems in the area have been resolved. It was long known that a factor $2$ could be achieved in a single-pass using $O(n\polylog n)$ space in the insertion-only model \cite{FeigenbaumKMSZ05} but this has just been shown to be essentially optimal \cite{assadi2026semistreamingmatchingsinglepass}. For the dynamic model where edges are both inserted and deleted, Assadi et al.~\cite{AssadiMatching26} showed that $O(\log \log n)$ passes are both sufficient (as mentioned above) and necessary in the standard $O(n\polylog n)$ space setting.

\section{Approximation Factor Analysis}

\begin{figure}[t]
  \centering
  \begin{subfigure}[b]{0.32\textwidth}
    \centering
    \includegraphics[width=\linewidth]{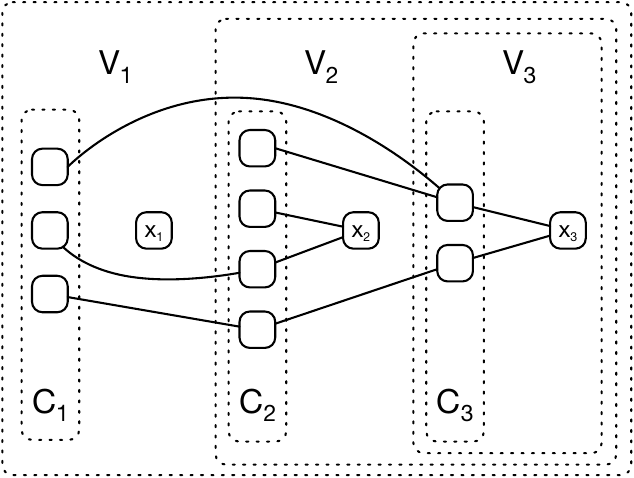}
    \caption{For each $v\in C_i$, the algorithm  picks an incident edge in $G[V_i]$. Call these $E_i$. Only these sampled edges are depicted. ~\\}
    \label{fig:sub1}
  \end{subfigure}\hfill
  \begin{subfigure}[b]{0.32\textwidth}
    \centering
    \includegraphics[width=\linewidth]{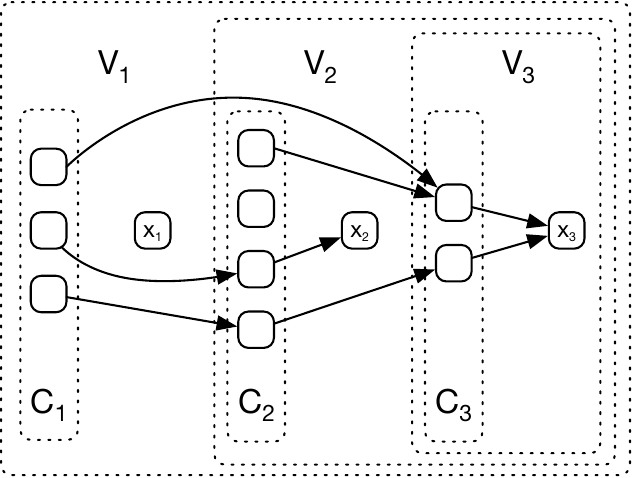}
    \caption{Sampled edges are filtered based on degree and random subsampling. $F_i$ is the set of remaining edges where these edges are directed away from $C_i$.}
    \label{fig:sub2}
  \end{subfigure}\hfill
  \begin{subfigure}[b]{0.32\textwidth}
    \centering
    \includegraphics[width=\linewidth]{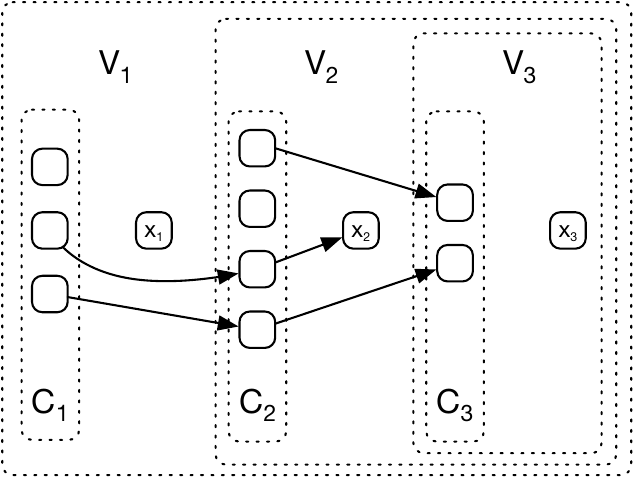}
    \caption{$S_i$ is determined by removing all edges in $F_i$ that point to the same vertices as other edges in $\cup_{j\geq i} F_j$. The remaining edges are vertex disjoint directed paths.}
    \label{fig:sub3}
  \end{subfigure}

  \caption{The algorithm defines  sequences $(x_i)_{i\geq 1}$, $(C_i)_{i\geq 1}$,  $(E_i)_{i\geq 1}$, and $(V_i)_{i\geq 1}$. The sequences $(F_i)_{i\geq 1}$ and  $(S_i)_{i\geq 1}$  are introduced in the analysis.}
  \label{fig:three-subfigs}
\end{figure}

The idea of the proof is to argue that within the sampled edges $\cup_i E_i$ there is a large set of edges forming vertex disjoint paths, and hence a  matching whose size is at least half the total number of edges in these paths. For the sake of analysis, it will be convenient to direct all edges in $G[V_i]$: we direct edges away from their higher degree endpoint (as measured in $G[V_i]$) with ties broken by an arbitrary total ordering of the vertices of $V_i$, i.e., 
\[ D_i = \{uv: \{u,v\}\in G[V_i], d_i(u)>d_i(v) \mbox{ or } (d_i(u)=d_i(v) \mbox{ and }u>v)\} \] 
where $d_i(\cdot )$ is  the degree in the induced subgraph $G[V_i]$.
We next restrict our attention to edges in $D_i$ corresponding to the sampled edges $E_i$ whose tail is in $C_i$. We independently subsample the elements of this set to create $F_i$, i.e., 
\[
F_i \leftarrow  \mbox{ subsample elements of  $\{uv \in D_i: u\in C_i, v=f(u)\}$  with probability $p$} \ .
\]
Lastly, we consider the subset 
\[S_i=\{ab\in F_i: a'b\not \in \cup_{j\geq i} F_j \mbox{ for all $a'\neq a$}\}\] 
i.e., we drop any edge in $F_i$ that  points towards the same vertex as an edge in $F_j$ for some $j\geq i$.

\begin{lemma}
The edges of $\cup_{i\geq 1} S_i$ form a collection of vertex disjoint paths. 
\end{lemma}
\begin{proof}
Every endpoint of an edge in $\cup_{i\geq 1} S_i$ has  out-degree at most one (any vertex $a$ is included in at most one set $C_j$ and will be the tail of at most one edge in $E_j$)
and in-degree at most one. The fact the in-degree is at most one follows because if $ab\in S_i$ and $a'b\in S_k$ for some $i\leq k$ then $a'b\in S_k\subseteq F_k\subseteq \cup_{j\geq i} F_{j}$ and this contradicts $ab\in S_i$. Hence, we can deduce $\cup_{i\geq 1} S_i$ is a collection of vertex disjoint paths and cycles. However, cycles are not possible: for any  $a\in C_i$ and  $ab\in S_i$ then either $b$ terminates the path, or $b$ is in $C_j$ for some $j>i$ (and the path can not return to $C_i$), or $b$ is in $C_i$ but ($d_i(b)=d_i(a)$ and $b<a$) or $d_i(b)<d_i(a)$ (and the path can never return to $a$). 
\end{proof}

In  Section \ref{sec:lem1} we will prove the following lemma:

\begin{lemma}\label{lem:1}
$\E{|S_i|}\geq \frac{\gamma_p}{2}\cdot \E{|C_i|}$ where $\gamma_p:=p(1-p)^2$.
\end{lemma} 

Theorem \ref{thm:main} then follows:
\begin{equation} \label{eq:main}
\E{\mu(\cup_{i\geq 1} E_i)}
\geq \E{\sum_{i\geq 1} |S_i|/2}
\geq  \frac{\gamma_p}{4}\cdot \E{\sum_{i\geq 1} |C_i|}\geq \frac{\gamma_p}{4} \cdot  \mu(G)
\end{equation}
and then optimizing over $p$ to get $p=1/3$.
Specifically, the first inequality follows because $ab\in S_i$ implies $\{a,b\}\in E_i$ and the maximum matching in a collection of vertex disjoint paths is at least half the number of edges in these paths. The last inequality follows because the size of any matching is upper bounded by the size of any vertex cover, and  $\cup_i C_i$ is a vertex cover.

\subsection{Proof of Lemma \ref{lem:1}}\label{sec:lem1}
It suffices to show the relationship holds conditioned on any  $V_i$ since the lemma then follows by taking expectations over $V_i$. We henceforth treat the $V_i$ conditioning as implicit.

Let  $n_i$ and $m_i$ be the number of vertices and edges in the induced subgraph $G[V_i]$. First note that \[\E{|C_i|}=\E{d_i(x_i)}=2m_i/n_i \ .\] 
Next, fix an edge $ab\in D_i$. If we can show $\Pr[ab\in S_i]\geq \gamma_p/n_i$ then \[\E{|S_i|}\geq \gamma_p \cdot m_i/n_i = \frac{\gamma_p}{2} \E{|C_i|}\] by linearity of expectation.
To do this,  we start with the following definitions.
\begin{itemize}
\item Let $X$ be the event $ab\in F_i$.
\item For $j\geq i$, let $Y_j$ be the event there exists $a'\neq a$ such that $a'b\in F_j$.
\end{itemize}
 Then,
\begin{equation}\label{eq0}
\Pr[ab\in S_i]= \Pr[X \cap \overline{Y}_i \cap \overline{Y}_{i+1} \cap  \ldots ] 
= \Pr[X]  (1-\Pr[Y_i \mid X]) (1-\Pr[\cup_{j\geq i+1} Y_j \mid X, \overline{Y}_i])\ .
\end{equation}
We bound each term as follows: \begin{equation}\label{eq1}
\Pr[X] = \Pr[a\in C_i, b=f(a)] \cdot p= d_i(a)/n_i \cdot 1/d_i(a) \cdot p=p/n_i \ .
\end{equation}
Note that for any edge $a'b\in D_i$ with $a'\neq a$,
 \begin{eqnarray*}
 \Pr[a'b \in F_i \mid X] &\leq & \Pr[a'\in C_i \mid X] \cdot \Pr[f(a')=b \mid X, a'\in C_i] \cdot p 
 \leq  
 1 \cdot 1/d_i(a') \cdot p
 \leq p/d_i(b)
 \end{eqnarray*}
even if $X$ implies $a'\in C_i$. 
Using this observation and the union bound,
\begin{equation}\label{eq2}
\Pr[Y_i \mid X]
\leq \sum_{a' : a'b \in D_i, a'\neq a} \Pr[a'b \in F_i  \mid X]
\leq p
\end{equation}
In the remainder of the proof, we show
\begin{equation}\label{eq3}
\Pr[\cup_{j\geq i+1} Y_j \mid X, \overline{Y}_i]\leq p \ .
\end{equation} 
Substituting \eqref{eq1}, \eqref{eq2}, and \eqref{eq3}, into \eqref{eq0} establishes $\Pr[ab\in S_i]\geq \gamma_p/n_i$ as required. 

\subsubsection{Establishing Equation \eqref{eq3}} \label{eee}
The approach is to show that it is likely that $b$ gets removed from the graph before an edge is added to some $F_j$ that points to $b$. To formalize the argument, we introduce some notation and terminology. Say that round $j$ is \emph{bad} if $a'b\in F_j$ for some $a'\neq a$ and \emph{good} otherwise. Say round $j$ is \emph{terminating} if $b\in C_j$.

\begin{itemize}
\item For $j\geq i$, let \[\cG_j=X\cap \overline{Y}_i \cap \ldots\cap \overline{Y}_{j}\]
i.e., $ab\in F_i$ and rounds $i, \ldots, j$ were good. 
\item 
For $j\geq i+1$, let \[\cB_j=\cG_{j-1} \cap Y_j\] i.e., $ab\in F_i$, rounds $i, \ldots, j-1$ were good but round $j$ was bad. 
\item For $j\geq i+1$, let 
\[\cT_j= \cG_j \cap \{b\in C_j\}\]
i.e., $ab\in F_i$, rounds $i, \ldots, j$ were good and round $j$ was terminating.
\end{itemize}

The events $\cB_{i+1}, \cB_{i+2},\ldots, \cT_{i+1}, \cT_{i+2}, \ldots $ can easily be shown to be mutually disjoint; the only non-trivial thing to note is that for $j<k$, we have that  $\cT_j$ and $\cB_{k}$ are disjoint because $\cT_j$  implies $b$ will not be present in $V_k$.
Set 
\[\alpha_j=\Pr[\cB_j \mid \cG_i]\qquad  \mbox{ and  } \qquad \beta_j=\Pr[\cT_j \mid \cG_i] \ . \] 
These disjointness properties imply  $\sum_{j\geq i+1} \alpha_j +\sum_{j\geq i+1} \beta_j \leq 1$ and so, if we can argue $\alpha_j \leq \frac{p \beta_j}{1-p}$, we can deduce
\begin{eqnarray*}
\Pr[\cup_{j\geq i+1} Y_j \mid X, \overline{Y}_i]= \Pr[ \cup_{j\geq i+1} Y_j \mid \cG_i ]
= \sum_{j\geq i+1} \alpha_j
\leq \frac{p}{1-p} \cdot  \sum_{j\geq i+1} 
\beta_j
\leq \frac{p}{1-p} \cdot (1-\sum_{j\geq i+1} \alpha_j)
 \end{eqnarray*}
 and rearranging establishes  \eqref{eq3}. 
 To prove $\alpha_j \leq \frac{p \beta_j}{1-p}$, note that for any non-empty $V_j$ 
 \begin{eqnarray*}
\Pr[\cB_j \mid V_j,\cG_{j-1}] 
\leq \sum_{a': a'b \in D_j} \Pr[a'b \in F_j \mid V_j,\cG_{j-1}] 
& \leq & \sum_{a': a'b \in D_j} \frac{d_j(a')}{n_j} \cdot \frac{p}{d_j(a')} \\
 &\leq  & {p \, d_j(b)}/{n_j}  \\
 &=&  p \Pr[b\in C_j \mid V_j,\cG_{j-1}] \\
& \leq & p \left (\Pr[\cT_j \mid V_j,\cG_{j-1}] + \Pr[\cB_j \mid V_j,\cG_{j-1}] \right )
\end{eqnarray*}
where, if $b\not \in V_j$ then we treated $d_j(b)$ as zero.
Averaging over $V_j$ and rearranging gives
\[
\Pr[\cB_j \mid \cG_{j-1}] \leq \frac{p}{1-p} \Pr[\cT_j \mid \cG_{j-1}] 
\]
and so
\[
\alpha_j = \Pr[\cB_j \mid \cG_{j-1}] \cdot \Pr[\cG_{j-1} \mid \cG_{i}]
\leq \frac{p}{1-p} \Pr[\cT_j \mid \cG_{j-1}] \cdot \Pr[\cG_{j-1} \mid \cG_i]
= \frac{p}{1-p} \beta_j .
\]

\section{Digression}\label{sec:dd}
A slight modification of the proposed algorithm would be to not just  remove $\{x_i\}$ and $C_i$ at the end of the $i$th step, but to also remove the set $\cup_{v\in C_i} f(v)$. This would ensure  $\cup_i E_i$ has a matching of size $\sum_i \mu(E_i)$ and the analysis could be further simplified  (specifically it would remove  the need for Section \ref{eee}). With this approach we have to slightly modify the analysis to take into account that $\cup_{v\in C_i} f(v)$ needs to be added to the vertex cover but this is doable. However, it seems unlikely that this modification of the algorithm could be implemented in $O(\log \log n)$ passes in the dynamic graph stream model. We share this observation to acknowledge the possibility that  the reader may be interested in another computational model other than the dynamic graph stream model.

\section*{Acknowledgements} ChatGPT and Claude were used to proofread drafts of the paper and perform a literature search on greedy maximal independent set algorithms. All ideation and proofs were by the human author.
The author takes full responsibility for the contents of the paper. 

{\small
\bibliographystyle{abbrv}
\bibliography{matchings}
}
\end{document}